\documentclass[11pt,a4paper]{article}

\usepackage{graphicx}
\usepackage{subfig}
\usepackage{mathptmx}
\usepackage{amsfonts}
\usepackage{amsmath}
\usepackage{amsthm}
\usepackage{amssymb}
\usepackage{xspace}
\usepackage[hyphens]{url}
\usepackage{hyperref}

\newtheoremstyle{mystyle}{}{}{\itshape}{}{\bfseries}{.}{5 pt}{\thmname{#1}\thmnumber{ #2}\thmnote{ {\bfseries(#3)}}}
\theoremstyle{mystyle}
\newtheorem{theorem}{Theorem}

\newtheorem{corollary}[theorem]{Corollary}

\newtheorem{conjecture}{Conjecture}
\newtheorem{metatheorem}{Metatheorem}

\numberwithin{equation}{section} 

\newcommand{\complexityclass}[1]{\textbf{#1}\xspace}
\newcommand{\computproblem}[1]{\textsc{#1}\xspace}
\renewcommand{\P}{\complexityclass{P}}
\newcommand{\NP}{\complexityclass{NP}}
\renewcommand{\L}{\complexityclass{L}}
\newcommand{\NL}{\complexityclass{NL}}
\newcommand{\PSPACE}{\complexityclass{PSPACE}}
\newcommand{\NPSPACE}{\complexityclass{NPSPACE}}
\newcommand{\EXP}{\complexityclass{EXP}}

\newcommand{\PTSAT}{\computproblem{Planar 3-SAT}}
\newcommand{\HAM}{\computproblem{Hamiltonian Cycle}}
\newcommand{\MCV}{\computproblem{Monotone Circuit Value}}
\newcommand{\QBF}{\computproblem{True Quantified Boolean Formula}}
\newcommand{\UCON}{\computproblem{Undirected Connectivity}}
\newcommand{\DCON}{\computproblem{Directed Connectivity}}

\begin{document}

\title{\textbf{Gaming is a hard job, but someone has to do it!}}
\author{
        \textsc{Giovanni Viglietta}\\
        Carleton University, Ottawa, Canada\\
        \normalsize
        \texttt{viglietta@gmail.com}
}
\date{\today} 

\maketitle

\begin{abstract}
We establish some general schemes relating the computational complexity of a video game to the presence of certain common elements or mechanics, such as destroyable paths, collectible items, doors opened by keys or activated by buttons or pressure plates, etc. Then we apply such ``metatheorems'' to several video games published between 1980 and 1998, including Pac-Man, Tron, Lode Runner, Boulder Dash, Deflektor, Mindbender, Pipe Mania, Skweek, Prince of Persia, Lemmings, Doom, Puzzle Bobble~3, and Starcraft. We obtain both new results, and improvements or alternative proofs of previously known results.
\end{abstract}

\section{Introduction}\label{s1}

This work was inspired mainly by the recent papers on the computational complexity of video games by Fori\v sek~\cite{platform} and Cormode~\cite{lemmings}, along with the excellent surveys on related topics by Kendall et al.~\cite{survey1} and Demaine et al.~\cite{survey2,hearn}, and may be regarded as their continuation on the same line of research.

Our purpose is to single out certain recurring features or mechanics in a video game that enable general reduction schemes from known hard problems to the games we are considering. To this end, in Section~\ref{s2} we produce several \emph{metatheorems} that will be applied in Section~\ref{s3} to a wealth of famous commercial video games, in order to automatically establish their hardness with respect to certain computational complexity classes (with a couple of exceptions).

Because most recent commercial games incorporate Turing-equivalent scripting languages that easily allow the design of undecidable puzzles as part of the gameplay, we will focus primarily on older, ``scriptless'' games. Our selection includes games published between 1980 and 1998, presented in alphabetical order for better reference. Not every game will be rigorously explained in all its aspects and details, but at least the game elements that are relevant to our proofs will be introduced, so that any casual player will promptly recognize them and readily understand our constructions.

It is clear that, in order to meaningfully apply the standard computational complexity tools, a suitable \emph{generalization} of each game must be considered. Since classic video games typically include only a finite set of levels, whose complexity is merely a constant, a way must be devised to automatically generate a class of infinitely many new levels of increasing size. Deciding which game elements are ``scalable'' and which are not is ultimately a matter of taste and common sense: when designing a generalization of a well-known game, one should remain as faithful as possible to the feeling and mechanics of the original version. For example, in a typical platform game, the number of platforms and the number of hazards in a level may increase as the level size grows. In contrast, the maximum height of a jump and the enemy AI should remain unchanged, as they are more inherent aspects of the game.

It is generally acknowledged that single-player games that are humanly ``interesting'' to play are complete either for \NP or for \PSPACE (for an introduction to general computational complexity theoretic concepts and classes, refer to~\cite{papadimitriou}). \NP-complete games feature levels whose solution demands some degree of ingenuity, but such levels are usually solved within a polynomial number of ``manipulations'', and the challenge is merely to find them. In contrast, the additional complexity of a \PSPACE-complete game seems to reside in the presence of levels whose solution requires an exponential number of manipulations, and this may be perceived as a nuisance by the player, as it makes for tediously long playing sessions.

Several open problems remain for further research: whenever only the hardness of a game is proved with respect to some complexity class, the obviously implied question is whether the game is also complete for that class. Moreover, different variations of each game may be studied, obtained for instance by further restricting the set of game elements used in our hardness proofs. Indeed, the computational complexity of a game is expected to dramatically drop if some ``critical'' elements are removed from its levels. It is interesting to study the ``complexity spectrum'' of a game, as a function of the game parameters that we set. This has been done to some extent for the game of Lemmings, by different authors, as partly documented in Section~\ref{s3}.

A conference version of this paper has appeared at FUN 2012~\cite{games}.

\section{Metatheorems}\label{s2}

More often than not, games allow the player to control an \emph{avatar}, either directly or indirectly. In some circumstances, an avatar may be identified within the game only through some sort of artifice or abstraction on the game mechanics. Throughout Section~\ref{s2}, we will stipulate that the player's actions involve controlling an avatar, and that the elements of the game may be freely arranged in a plane lattice, or a higher dimensional space. At the very least, the set of game elements includes \emph{walls} that cannot be traversed by the avatar, and can be arranged to form rooms, paths, etc.

In general, a problem instance will be a ``level'' of a given game. The description of a level includes the position of every relevant game element, such as walls, items, the avatar's starting location, etc. The question is always whether or not a given level can be ``solved'' under certain conditions, such as losing no lives, etc. The exact definition of ``solvability'' is highly game-dependent, and can range from reaching an exit location, to collecting some items, to killing some enemies, to surviving for a certain time, etc.

All the \emph{metatheorems} that follow yield hardness results under the assumption that certain game elements are present in a given game. These are not to be intended as ``black boxes'', as regular \emph{theorems} would be, but rather as ``frameworks''. Indeed, we will not always be able to apply the statement of a metatheorem to a particular game without keeping in mind the actual proof of the metatheorem, and the underlying construction enabling the reduction. As it turns out, in order to apply a metatheorem in a non-trivial way, we may need to use certain game elements having very complex behaviors, which serve our purposes only when arranged in some special ways. In order to make sure that our constructions work as intended, we may have to access the full proof of the metatheorem, and exploit some of its features at a ``lower level''. To avoid all this, we would have to strengthen the metatheorems' statements by adding so many details about the actual reduction constructions that most of their appeal would be lost. Here we opt for shorter metatheorem statements, but as a drawback we will have to refer to their proofs from time to time, when invoking them.

\subsection{Location traversal and single-use paths}

A game is said to exhibit the \emph{location traversal} feature if the level designer can somehow force the player's avatar to visit several specific game locations, arbitrarily connected together, in order to beat the level. Although every location must be visited at least once, the avatar may visit them multiple times and in any order. However, the first location is usually fixed (starting location), and sometimes also the last one is (exit location). An example of location traversal is the \emph{collecting items} feature discussed in~\cite{platform}: a certain number of items are scattered across different locations, and the avatar's task is to collect them all.

The \emph{single-use paths} feature is the existence of configurations of game elements that act as paths connecting two locations, which can be traversed by the avatar at most once. A typical example are \emph{breakable tiles}, which disappear as soon as the avatar walks on them.

\begin{metatheorem}\label{m1}
Any game exhibiting both location traversal (with or without a starting location or an exit location) and single-use paths is \NP-hard.
\end{metatheorem}
\begin{proof}
We give a straightforward reduction from \HAM, which is \NP-complete even for undirected 3-regular planar graphs~\cite{ham1,papadimitriou}. Construct a plane embedding of a given 3-regular graph $G$ (perhaps an orthogonal embedding, if needed) with an additional vertex $u$ dangling from a distinguished vertex $v$. Then we convert such embedding into a valid level, by implementing each vertex as a location that must be visited by the avatar, and each edge as a single-use path. The starting location is placed in $v$ and, if an exit location is required, it is placed in $u$.

Clearly, the last vertex the avatar must visit is $u$, because it has only one incident edge. Moreover, each vertex except $v$ can be visited at most once: recall that $G$ is 3-regular, hence reaching a given vertex $w\not \in \{u,v\}$ for the first time implies the consumption of one of its incident edges. Then, leaving $w$ consumes another incident edge, and reaching it a second time consumes the third incident edge. At this point, there is no way for the avatar to leave $w$, and therefore no way to reach the last vertex $u$. As for the starting vertex $v$, the incident edges are initially four, and one is immediately consumed. The second time the avatar reaches $v$, it must necessarily proceed to $u$, for otherwise $u$ would become forever unreachable. It follows that the level is solvable if and only if the player can find a walk starting from $v$, touching every vertex (except $u$) exactly once, reaching $v$ again, and then terminating in $u$. This is possible if and only if $G$ contains a Hamiltonian cycle.
\end{proof}

It is easy to see that \NP-hardness is the best we can achieve given the hypotheses of Metatheorem~\ref{m1}.

\begin{corollary}
There exists an \NP-complete game exhibiting location traversal and single-use paths.
\end{corollary}
\begin{proof}
Consider the game $\mathcal G$ played on an undirected graph, in which some distinguished edges implement single-use paths, and some distinguished vertices must be visited by the avatar in order to win. Then, by Metatheorem~\ref{m1} $\mathcal G$ is \NP-hard, while a certificate for $\mathcal G$ is an injective sequence of distinguished vertices and distinguished edges.
\end{proof}

Also notice that both assumptions of Metatheorem~\ref{m1} are required: removing either of them from the above game $\mathcal G$ reduces it to determining if two vertices in a graph are connected, which is solvable in logarithmic time (see~\cite{ucon}).

As Section~\ref{s3} testifies, Metatheorem~\ref{m1} has a wide range of applications, and it tends to yield game levels that are more ``playable'' than those resulting from the somewhat analogous~\cite[Metatheorem~2]{platform}, which rely on a tight time limit to traverse a grid graph. Additionally,~\cite[Metatheorem~2]{platform} is prone to design complications in ``anisotropic'' games, in which the avatar moves at different speeds in different directions, for instance due to gravity effects.

\subsection{Tokens and toll roads}

We consider now another type of game mechanics: \emph{tokens} and \emph{toll roads}. Tokens are items that can be carried by the avatar, and \emph{toll roads} are special paths connecting two locations. Whenever the avatar traverses a toll road, it must ``spend'' a token that it is carrying. If the avatar is carrying no token, then it cannot traverse a toll road.

We distinguish between two types of tokens: \emph{collectible} tokens, which may be placed by the game designer at specific locations and can be picked up by the avatar, and \emph{cumulative} tokens, any number of which can be carried around by the avatar at the same time. Section~\ref{s3} will offer some examples of different types of tokens: for instance, Pac-Man features \emph{power pills}, which may be regarded as collectible tokens that are not cumulative.

\begin{metatheorem}\label{m1b}
A game is \NP-hard if either of the following holds:
\begin{enumerate}
\item[\textbf{\emph{(a)}}] The game features \emph{collectible} tokens, toll roads, and location traversal.
\item[\textbf{\emph{(b)}}] The game features \emph{cumulative} tokens, toll roads, and location traversal.
\item[\textbf{\emph{(c)}}] The game features \emph{collectible cumulative} tokens, toll roads, and the avatar has to reach an exit location.
\end{enumerate}
\end{metatheorem}
\begin{proof}
Once again, we give a reduction from \HAM for all three parts of the metatheorem, varying it slightly depending on our hypotheses. For part~(a), given an undirected 3-regular planar graph $G$, we construct an embedding as described in the proof of Metatheorem~\ref{m1}, and we implement each vertex as a location that has to be traversed by the avatar. Each edge is then implemented as a toll road, and one collectible token is placed in each vertex, except for the final vertex $u$, where we place no token, and the starting vertex $v$, where we place two tokens.

Notice that, if $G$ has $n$ vertices, there are exactly $n+1$ locations that the avatar must visit, and $n+1$ tokens in the level. Therefore, any feasible traversal of the level starts from $v$ and has length at most $n+1$. If the traversal must reach all locations, then at most one location may be visited twice. Moreover, $v$ must be visited at least a second time, because it is the only neighbor of $u$. As a consequence, a valid traversal of the level must start from $v$, visit every other location except $u$ exactly once, return in $v$, and end in $u$. It follows that, if $G$ has no Hamiltonian cycle, then the level is unsolvable.

Conversely, let us assume that $G$ has a Hamiltonian cycle, and let us show that the level is solvable. The avatar can traverse $G$, starting from $v$ and ending in $v$ again, along a Hamiltonian cycle, and finally it can reach $u$ and solve the level. This traversal is valid even if tokens are not cumulative: upon reaching a new location, the avatar collects one new token and immediately spends it in a toll road. Likewise, when $v$ is reached for the second time, the second token is collected, and it is immediately spent to reach $u$.

The construction for part~(b) is the same, but instead of scattering $n+1$ tokens throughout the level (where $n$ is the number of vertices of $G$), we assume that the avatar already carries $n+1$ tokens as the game starts. Then a similar reasoning applies: exactly one location may be visited twice, which must be $v$ because it is the starting location and the only neighbor of $u$.  Therefore, $u$ must be the last location to be visited, and the level is solvable if and only if $G$ has a Hamiltonian cycle.

For part~(c), we further modify the previous proof as follows: we construct the same embedding of $G$, and we place two tokens in every location, except in $u$, where we place no token. Then we implement each edge as a toll road, except the edge between $v$ and $u$, which is implemented as a sequence of $n$ toll roads. The starting location is $v$ again, and the exit location is $u$. The avatar carries no token as the game starts.

There are $2n$ tokens in the level, and $n$ of them must be used to travel from $v$ to $u$, so at most $n$ more tokens may be spent in other toll roads. Every time a toll road is traversed, one token is gained if a new location is reached (one token is spent and two are found), and one token is lost if an already visited location is reached. It follows that the player must find a walk in $G$ that starts and ends in $v$, traverses at most $n$ edges and visits $n$ different vertices. This is equivalent to finding a Hamiltonian cycle in $G$. (Observe that the location traversal feature has been obtained here as a by-product of our construction, without being an explicit requirement.)
\end{proof}

\NP-hardness is the best complexity achievable under the hypotheses of Metatheorem~\ref{m1b}, in each of the three cases.

\begin{corollary}
There exists an \NP-complete game featuring collectible cumulative tokens and toll roads, in which the avatar has to reach an exit location.
\end{corollary}
\begin{proof}
Consider the game played on a graph in which some distinguished edges implement toll roads, each vertex may contain some collectible cumulative tokens, and one distinguished vertex is the exit location. Indeed, a certificate for this game is simply an injective sequence of toll roads, because we may assume that the avatar always collects all the tokens it can reach without traversing toll roads, and therefore no toll read ever has to be traversed twice.
\end{proof}

\subsection{Doors and keys}

A \emph{door} is a game element that can be open or closed, and may be traversed by the avatar if and only if it is open. A \emph{key} is a type of token that can be used by the avatar to open a closed door, upon contact. Any key can open any door, but a key is ``consumed'' as soon as it is used. Hence, the key-door paradigm is somewhat similar to the token-toll road one, with the difference that a door opened by a key remains open and can be traversed several times afterwards without consuming new keys.

We distinguish again between \emph{collectible} keys, which can be found by the avatar and picked up, and \emph{cumulative} keys, any number of which can be carried at the same time. Many examples of keys are found in platform games and adventure games. In Section~\ref{s3}, we will show how the Lemmings game features cumulative keys that are not collectible, although this will be established through non-trivial abstractions on the game mechanics.

To state the next result, which is an analogous of Metatheorem~\ref{m1b} for the key-door paradigm, we further need to introduce the concept of \emph{one-way path}, which is a path that can be traversed by the avatar in one specific direction only.

\begin{metatheorem}\label{m1c}
A game is \NP-hard if it contains doors and one-way paths, and either of the following holds:
\begin{enumerate}
\item[\textbf{\emph{(a)}}] The game features \emph{collectible} keys and location traversal.
\item[\textbf{\emph{(b)}}] The game features \emph{cumulative} keys and location traversal.
\item[\textbf{\emph{(c)}}] The game features \emph{collectible cumulative} keys and the avatar has to reach an exit location.
\end{enumerate}
\end{metatheorem}
\begin{proof}
We reduce from \HAM, which is \NP-complete even for directed planar graphs whose vertices have one incoming edge and two outgoing edges, or two incoming edges and one outgoing edge~\cite{ham2}. All three parts of our proof are based on the same construction: given one such directed graph $G$ on $n$ vertices, we pick a vertex $v$ with indegree two and outdegree one (which exists), and we attach to it a new outgoing edge, ending in a new vertex $u$. Then we construct a plane embedding of this graph (maybe an orthogonal embedding), substituting each vertex with a game location, and each directed edge with a one-way path. $v$ will be the avatar's starting location and, if an exit location is required, it is placed in $u$ ($u$ must be the final location anyway, because it has no outgoing edges). Moreover, we place a closed door in each one-way path, except in the path between $v$ and $u$.

To prove part~(a), place one key in each location, except $u$, and assume that the avatar must traverse every location (the last of which would be $u$). After collecting the first key in $v$, every time a new location (except $u$) is reached, one key is used to open a door, and one new key is found. On the other hand, as soon as an already visited location $w$ is reached, the only key in the avatar's possession is lost and no key is found, so afterwards the avatar is bound to traverse only paths with no door (hence $(v,u)$) or with an already opened door. As a consequence, the level is solved if and only if $w=v$ and every location except $u$ has been visited, which is possible if and only if $G$ has a Hamiltonian cycle.

For part~(b), we put no keys in the level, but we assume that the avatar already carries $n$ keys as the game starts. We assumed that each location must be visited, and therefore at least two of its incident paths' doors must be opened (unless the location is $u$). Hence, $n$ doors must be opened in total, and all the keys must be used. On the other hand, if all the three incident paths' doors of a location are opened, and all the locations are visited, a straightforward double counting argument shows that at least $n+1$ keys have been used, which is unfeasible. Therefore, the avatar must follow a Hamiltonian cycle of $G$ starting and ending in $v$, and then visit $u$.

Finally, for part~(c), we place two keys in each location, except in $u$, and we place $n$ doors in the path between $v$ and $u$. No key is carried by the avatar as the game starts. Hence, the avatar must visit some locations to collect keys, return to $v$ with at least $n$ keys, and reach the exit in $u$. The same double counting argument used for part~(b) reveals that, if $k\leqslant n$ distinct locations are visited before reaching $u$, then at at least $k$ doors must be opened. In particular, exactly $k$ doors are opened if and only if a cycle of $G$ is followed. Because visiting $k$ distinct locations allows to collect exactly $2k$ keys, the only way to return in $v$ with $n$ keys is to follow a cycle of length $k=n$, i.e., a Hamiltonian cycle.
\end{proof}

\begin{corollary}
There exists an \NP-complete game featuring doors, collectible cumulative keys, one-way paths, location traversal, in which the avatar has to reach an exit location.
\end{corollary}
\begin{proof}
Consider the game played on a graph in which some edges may implement doors or one-way paths, some distinguished vertices must be traversed by the avatar, some vertices contain collectible cumulative keys, and one vertex is the exit location. A certificate for this game is an injective sequence of distinguished vertices, vertices containing keys, and edges containg doors. Indeed, all the keys contained in a vertex can be taken as soon as the vertex is reached, and any open door becomes a regular path and does not have to be opened a second time. Hence the game is in \NP, and by Metatheorem~\ref{m1c} it is \NP-complete.
\end{proof}

\subsection{Doors and pressure plates}

There are other ways to modify a door's status, such as pushing a \emph{pressure plate}. A pressure plate is a floor button that is operated whenever the avatar steps on it, and its effect may be either the opening or the closure of a specific door. Each pressure plate is connected to just one door, and each door may be controlled by at most two pressure plates (one opens it, one closes it). Of course, all our hardness results will hold in the more general scenario in which any number of doors is controlled by the same pressure plate, or any number of pressure plates control the same door.

In~\cite[Metatheorem~3]{platform}, Fori\v sek shows (with a different terminology) that a game is \NP-hard if the avatar has to reach an exit location, and the game elements include one-way paths, doors and pressure plates (or 1-buttons, see the next subsection) that can open doors. In the following metatheorm, we further explore the capabilities of pressure plates.

We say that a game allows \emph{crossovers} if there is a way to prevent the avatar from switching between two crossing paths. Some 2-dimensional games natively implement crossovers through bridges or tunnels. In some other games, crossovers can be simulated through more complicated gadgets.

\begin{metatheorem}\label{m2}
If a game features doors and pressure plates, and the avatar has to reach an exit location in order to win, then:
\begin{enumerate}
\item[\textbf{\emph{(a)}}] Even if no door can be closed by a pressure plate, and if crossovers are allowed, then the game is \P-hard.
\item[\textbf{\emph{(b)}}] Even if no two pressure plates control the same door, the game is \NP-hard.
\item[\textbf{\emph{(c)}}] If each door may be controlled by two pressure plates, then the game is \PSPACE-hard.
\end{enumerate}
\end{metatheorem}
\begin{proof}
To prove part~(a), we give a \L-reduction from \MCV \cite{papadimitriou}. OR and AND gates are implemented as in Figures~\ref{f1a} and~\ref{f1b}, the starting location is connected to all true input literals, and the exit is located on the output. It is easy to check that the output of an OR gate can be reached by the avatar if and only if at least one of its two input branches is. Similarly, both input branches of an AND gate must be reached by the avatar in order for doors $a$ and $b$ to be opened and allow access to the output. The doors $x$ and $y$ in the AND gate prevent the avatar from walking from one input branch to the other through the center of the gate, in case only one input branch is reachable. Clearly, the exit is eventually reachable if and only if the output of the circuit is true.

\begin{figure}[htbp]
\centering
\subfloat[OR gate.]{
	\label{f1a}
	\includegraphics[scale=2.5]{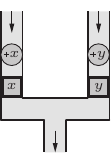}}\qquad\qquad
\subfloat[AND gate.]{
	\label{f1b}
	\includegraphics[scale=2.5]{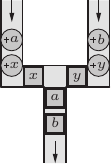}}\qquad\qquad
\subfloat[Single-use path.]{
	\label{f1c}
	\includegraphics[scale=2.5]{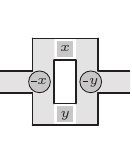}}
\caption{Several gadgets implemented with doors and pressure plates. A pressure plate that is labeled $+x$ (respectively,~$-x$) opens (respectively,~closes) the unique door labeled $x$, and a door is initially open (respectively,~closed) if its border is white (respectively,~black).}
\label{f1}
\end{figure}

For part~(b), observe that we can implement single-use paths as shown in Figure~\ref{f1c}: in order to traverse the gadget, the avatar must walk on both pressure plates, thus permanently closing both doors. Since we can also enforce location traversal by blocking the exit with several closed doors, which may be opened via as many pressure plates positioned in every location, we may indeed invoke Metatheorem~\ref{m1}.

Finally, to prove~(c), we implement a reduction framework from \QBF, sketched in Figure~\ref{f2}. A given fully quantified Boolean formula $\exists x \forall y \exists z \cdots \varphi(x, y, z, \cdots)$, where $\varphi$ is in 3-CNF, is translated into a row of \emph{Quantifier gadgets}, followed by a row of \emph{Clause gadgets}, connected by several paths.

\begin{figure}[htbp]
\centering
\includegraphics[width=\linewidth]{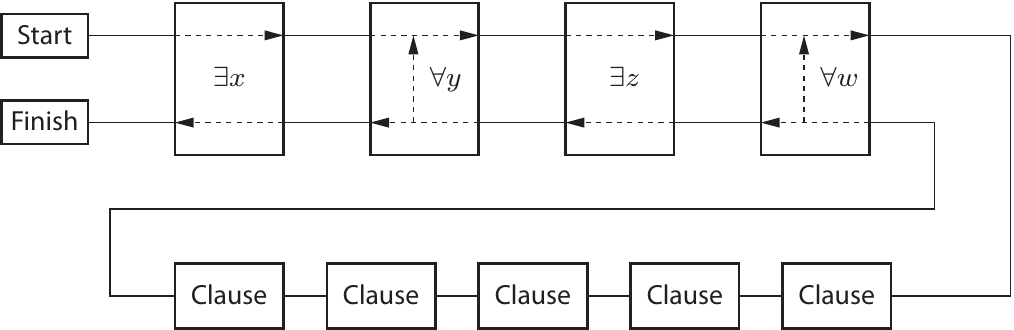}
\caption{Sketch of the PSPACE-hardness framework.}
\label{f2}
\end{figure}

Traversing a Quantifier gadget at any time sets the truth value of the corresponding Boolean variable. On the other hand, each Clause gadget can be traversed if and only if the corresponding clause of $\varphi$ is satisfied by the current variable assignments. Whenever traversing an Existential Quantifier gadget, the player can choose the truth value of the corresponding variable. On the other hand, the first time a Universal Quantifier gadget is traversed, the corresponding variable is set to true.

When all variables are set, the player attempts to traverse the Clause gadgets. If the player succeeds, he proceeds to the ``lower parts'' of the Quantifier gadgets, where he is rerouted to the last Universal Quantifier gadget in the sequence. The corresponding variable is then set to false, and $\varphi$ is ``evaluated'' again by making the player walk through all the Clause gadgets.

The process continues, forcing the player to ``backtrack'' several times, setting all possible combinations of truth values for the universally quantified variables, and choosing appropriate values for the existentially quantified variables in the attempt to satisfy $\varphi$.

Finally, when all the necessary variable assignments have been tested and $\varphi$ keeps being satisfied, i.e., if the overall quantified Boolean formula is true, the exit becomes accessible, and the player may finish the level. Conversely, if the quantified Boolean formula is false, there is no way for the player to operate doors in order to reach the exit.

Next we show how to implement all the components of our framework using just doors and pressure plates.

Clause gadgets are straightforwardly implemented, as shown in Figure~\ref{f3}. There is a door for each literal in the clause, and the avatar may traverse the clause if and only if at least one of the doors is open.

\begin{figure}[htbp]
\centering
\includegraphics[scale=2.5]{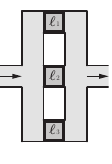}
\caption{Clause gadget implemented with doors.}
\label{f3}
\end{figure}

The Existential Quantifier gadget for variable $x$ is illustrated in Figure~\ref{f4}. $x_1$, $x_2$, etc.\ (respectively, $\overline{x_1}$, $\overline{ x_2}$, etc.) denote the positive  (respectively, negative) occurrences of $x$ in the clauses of $\varphi$.

\begin{figure}[htbp]
\centering
\includegraphics[scale=2.5]{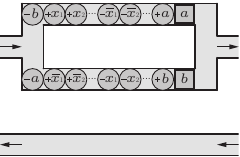}
\caption{Existential Quantifier gadget implemented with doors and pressure plates. $x_i$ (respectively, $\overline{x_i}$) denotes the $i$-th occurrence of literal $x$ (respectively, $\neg x$) in $\varphi$.}
\label{f4}
\end{figure}

When traversing the upper part of the gadget from left to right, the player must choose one of the two paths, thus setting the truth value of $x$ to either true or false. This is done by appropriately opening or closing all the doors corresponding to occurrences of $x$ in $\varphi$. The doors labeled $a$ and $b$ prevent leakage between the two different paths of the Existential Quantifier gadget, enforcing mutual exclusion.

Finally, the lower part of the gadget is traversed from right to left when the player backtracks, and it is simply a straight path.

A Universal Quantifier gadget for variable $x$ is shown in Figure~\ref{f5}.

\begin{figure}[htbp]
\centering
\includegraphics[scale=2.5]{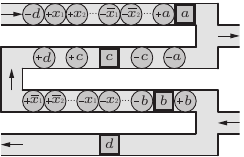}
\caption{Universal Quantifier gadget implemented with doors and pressure plates.}
\label{f5}
\end{figure}

When the avatar enters the gadget from the top left, door $d$ gets closed and variable $x$ is set to true. Then the avatar must exit to the top right, because door $c$ cannot be traversed from right to left.

When backtracking the first time, the avatar enters from the bottom right and, because door $d$ is still closed, it must take the upper path, thus setting variable $x$ to false. Incidentally, door $d$ gets opened and door $a$ gets closed, thus preventing leakage to the top left entrance, and forcing the avatar to exit to the top right again.

When backtracking the second time (i.e., when both truth values of $x$ have been tested), door $d$ is open and the avatar may finally exit to the bottom left. When done backtracking, the avatar will eventually enter this gadget again from the top left, setting $x$ to true again, etc. 

We note that, as a result of our constructions, each door is operated by exactly two pressure plates. For instance, the door labeled $x_1$, located in some Clause gadget, is opened and closed by exactly two pressure plates, both located in the Quantifier gadget corresponding to variable $x$.
\end{proof}

Observe that our Metatheorem~\ref{m2}.c is an improvement on~\cite[Metatheorem~4]{platform}, in that the \emph{long fall} feature (and thus the concept of gravity) is not used, and it works with a more restrictive model of doors: in~\cite{platform}, arbitrarily many pressure plates can act on the same door, while we allow just two.

As with previous metatheorems, we can prove that Metatheorem~\ref{m2}'s statement is the best possible given its hypotheses.

\begin{corollary}\label{cor2}
There exist games $\mathcal G_1$, $\mathcal G_2$, $\mathcal G_3$, featuring doors and pressure plates, in which the avatar has to reach an exit location, such that:
\begin{enumerate}
\item[\emph{(a)}] In $\mathcal G_1$, pressure plates can only open doors, crossovers are allowed, and $\mathcal G_1$ is \P-complete.
\item[\emph{(b)}] In $\mathcal G_2$, no two pressure plates control the same door, and $\mathcal G_2$ is \NP-complete.
\item[\emph{(c)}] In $\mathcal G_3$, each door may be controlled by two pressure plates, and $\mathcal G_3$ is \PSPACE-complete.
\end{enumerate}
\end{corollary}
\begin{proof}
We consider games played on graphs whose vertices may contain pressure plates, and whose edges may contain doors. Then, $\mathcal G_1$ belongs to \P because the set of vertices accessible to the avatar can only increase whenever a pressure plate is activated. As new pressure plates become accessible, the avatar immediately activates them, opening new doors, until either the exit becomes accessible, or no new pressure plates are discovered.

$\mathcal G_2$ is in \NP because no pressure plate can undo the effects of another pressure plate. Therefore, we may pretend that a pressure plate disappears as soon as it is activated. It follows that a certificate for this game is simply an injective sequence of pressure plates.

Finally, to see that $\mathcal G_3$ is in \PSPACE $=$ \NPSPACE (cf.~Savitch's theorem~\cite{papadimitriou}), it is sufficient to observe that the a level's \emph{state} can be stored in linear space, allocating one bit for each door and storing the position of the avatar in the graph. Then, a certificate is just a walk in the graph.
\end{proof}

Metatheorem~\ref{m2}.c has a wide range of straightforward applications: most first-person shooters (with the notable exception of Wolfenstein~3D), adventure games, and dungeon crawls are all \PSPACE-hard. This includes RPGs such as Dungeon Master, The Eye of the Beholder, and Lands of Lore, which natively implement doors operated by pressure plates. Similar mechanisms can be implemented also in the first-person shooter Doom and its sequels, via walkover lines and sector tags. In simple terms, whenever the player-controlled avatar crosses a certain line on the ground, a ``block'' somewhere in the level is moved to a predefined location, thus simulating the opening or closure of a door. All the point-and-click adventure games based on LucasArts' SCUMM engine, such as Maniac Mansion and The Secret of Monkey Island, as well as most Sierra's adventure games, easily fall in this category, too.

\subsection{Doors and buttons}

A \emph{button} is similar to a pressure plate, except that the player may choose whether to push it or not, whenever his avatar encounters one.

Games with buttons are in general not harder than games with pressure plates, because a pressure plate can trivially simulate a button, as Figure~\ref{f6a} shows. However, since the converse statement is not as clear, we will allow a single button to act on several doors, in contrast with pressure plates. A button acting on $k$ doors simultaneously is called a \emph{$k$-button}.

\begin{figure}[htbp]
\centering
\subfloat[Simulating a button for $\pm x$ with a pressure plate.]{
	\label{f6a}
	\includegraphics[scale=2.5]{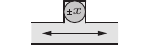}}\qquad\qquad\qquad
\subfloat[Single-use path implemented with doors and 2-buttons.]{
	\label{f6b}
	\includegraphics[scale=2.5]{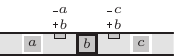}}
\caption{Button-related constructions. Small rectangles represent buttons, with hovering letters indicating their functions.}
\label{f6}
\end{figure}

We obtain an analogous of Metatheorem~\ref{m2} for buttons.

\begin{metatheorem}\label{m3}
If a game features doors and $k$-buttons, and the avatar has to reach an exit location in order to win, then:
\begin{enumerate}
\item[\textbf{\emph{(a)}}] If $k\geqslant 1$ and crossovers are allowed, then the game is \P-hard.
\item[\textbf{\emph{(b)}}] If $k\geqslant 2$, then the game is \NP-hard.
\item[\textbf{\emph{(c)}}] If $k\geqslant 3$, then the game is \PSPACE-hard.
\end{enumerate}
\end{metatheorem}
\begin{proof}
For part~(a), we mirror the proof of Metatheorem~\ref{m2}.a, by using $1$-buttons as opposed to pressure plates. Indeed, pressing a button to open a door is never a ``wrong'' move, if the goal is to reach the exit location.

For part~(b), we implement single-use paths as in Figure~\ref{f6b}: in order to open door $b$, one of the two buttons has to be pressed, thus permanently closing door $a$ or door $c$. Then we proceed as in the proof of Metatheorem~\ref{m2}.b.

Finally, for part~(c), we use the gadget in Figure~\ref{f7a} to simulate a generic pressure plate for $\pm x$: the only way to traverse the gadget from left to right (the other direction is symmetric) is to press the buttons as indicated in Figures~\ref{f7b},~\ref{f7c}, and~\ref{f7d}, incidentally activating also door $x$. Moreover, observe that, no matter how the six buttons are operated, there is no way to ``break'' the gadget by leaving its four doors in an open/closed state that is not the original one. Now, by simulating general pressure plates, we can apply the \PSPACE-hardness framework used for Metatheorem~\ref{m2}.c, concluding the proof. (Observe that our gadget can be further simplified: by inspecting the proof of Metatheorem~\ref{m2}.c, it is apparent that each pressure plate is traversed by the avatar in only in one direction. Hence, only three buttons are sufficient in our gadget, e.g., those that are used to traverse it from left to right.)
\end{proof}

\begin{figure}[htbp]
\centering
\subfloat[Starting configuration.]{
	\label{f7a}
	\includegraphics[scale=2.5]{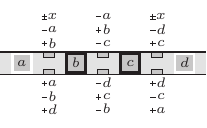}}\qquad\qquad
\subfloat[Approaching from the left, the circled button must be pressed to open door $b$. This also operates door $x$, as required.]{
	\label{f7b}
	\includegraphics[scale=2.5]{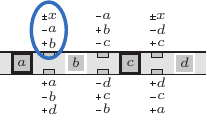}}\\
\subfloat[A second button must be pressed to open door $c$ and proceed.]{
	\label{f7c}
	\includegraphics[scale=2.5]{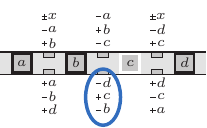}}\qquad\qquad
\subfloat[A third button must be pressed in order to open door $d$ and exit the gadget to the right. The gadget's final configuration is the same as the initial one.]{
	\label{f7d}
	\includegraphics[scale=2.5]{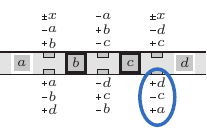}}
\caption{Simulating a pressure plate for $\pm x$ with 3-buttons.} \label{f7}
\end{figure}

We can show that Metatheorem~\ref{m3}.a and Metatheorem~\ref{m3}.c are tight.

\begin{corollary}
There exist games $\mathcal G_1$ and $\mathcal G_2$, featuring doors, in which the avatar has to reach an exit location, such that:
\begin{enumerate}
\item[\emph{(a)}] $\mathcal G_1$ features $1$-buttons and is \P-complete.
\item[\emph{(b)}] $\mathcal G_2$ features $k$-buttons, with $k\geqslant 3$, and is \PSPACE-complete.
\end{enumerate}
\end{corollary}
\begin{proof}
Once again, we consider games played on a graph, in which some vertices contain buttons, and some edges implement doors. To prove that $\mathcal G_1$ belongs to \P, observe that each 1-button either opens or closes a door. 1-buttons opening doors can be presses as soon as they become accessible, and 1-buttons closing doors may be ignored. Therefore the game is equivalent to that of Corollary~\ref{cor2}.a.

Similarly to the game of Corollary~\ref{cor2}.c, $\mathcal G_2$ belongs to \PSPACE $=$ \NPSPACE because the level's state can be stored in linear space, and a certificate is the sequence of the avatar's positions and the $k$-switches pressed.
\end{proof}

It remains an open problem to establish if also Metatheorem~\ref{m3}.b is tight.

\begin{conjecture}
There exists an \NP-complete game featuring doors and 2-buttons in which the avatar has to reach an exit location.
\end{conjecture}

It is easy to construct levels in which some 2-buttons have to be pressed more than once in order to reach the exit, therefore it is not trivially true that the sequence of buttons pressed is a polynomial certificate.

\section{Applications and further results}\label{s3}

In this section we apply the previous metatheorems to some well-known games. In some cases we will merely have to ``simulate'' all the required elements and mechanics with appropriate gadgets, and the reductions will immediately follow from metatheorems' statements. In other cases, the desired effects will be obtained by taking into account specific features of a metatheorem's proof, and noticing that the patterns we wish to construct have particular properties that, for the sake of brevity, are not mentioned in the actual statement of the metatheorem (cf.~Section~\ref{s2}'s introduction).

Most of the results we prove are new. For the results that were already known, we either provide simplified reductions, or reductions that use a different set of game elements. The following table relates all the games that we consider in this section with their complexities, indicating whether or not each result is new.

\begin{center}
\begin{tabular}{c|c|c}
\textbf{Game} & \textbf{Complexity} & \textbf{New}\\
\hline
Boulder Dash & \NP-hard & \\
Deflektor & \L & \checkmark \\
Lemmings & \NP-hard & \\
Lode Runner & \NP-hard & \checkmark \\
Mindbender & \NL-hard & \checkmark \\
Pac-Man & \NP-hard & \checkmark \\
Pipe Mania & \NP-complete & \checkmark \\
Prince of Persia & \PSPACE-complete & \\
Puzzle Bobble 3 & \NP-complete & \checkmark \\
Skweek & \NP-hard & \checkmark \\
Starcraft & \NP-hard & \checkmark \\
Tron & \NP-hard & \checkmark \\
\end{tabular}
\end{center}

Some of our reductions produce quite contrived levels or configurations, which are very unlikely to occur in the real games: for instance, a draw in Starcraft is somewhat rare, and a Tron configuration such as the one presented is definitely unnatural. For these games, designing reductions that preserve most of the relevant aspects of the gameplay remains a challenging open problem.

\subsection{Boulder Dash (First Star Software, 1984) is \NP-hard}

\paragraph{Game description.}
The game is similar to Sokoban, but with added gravity. The player-controlled avatar may push (but not pull) single boulders horizontally, excavate some special tiles, and must collect diamonds while avoiding monsters. When a certain amount of diamonds has been collected, an exit door appears, and the avatar has to reach it to beat the level. Gravity affects boulders and diamonds, but not the avatar or the monsters.

\paragraph{Complexity.}
A proof that ``pushing blocks in gravity'' is \NP-hard has been given by Friedman in 2002~\cite{gravity}, based on a rather involved reduction scheme and several gadgets that may be adapted to work with the slightly different ``physics'' of Boulder Dash (namely, if a boulder is on the ``edge'' of a pit, it falls down even if it is not pushed).

We give a much simpler proof that relies on Metatheorem~\ref{m1}. Location traversal is trivially implemented, due to the presence of diamonds that have to be collected: we place one diamond in each relevant location, forcing the player to visit them all. A single-use path gadget is illustrated in Figure~\ref{fbd}: when traversing the gadget in either direction for the first time, three boulders are pushed in the pits. On the second traversal attempt, the fourth boulder blocks the path.

Notably, we do not use diggable tiles or monsters in our reduction, although we do require diamonds.

\begin{figure}[htbp]
\centering
\subfloat[Initial configuration.]{
	\label{fbd1}
	\includegraphics[scale=0.75]{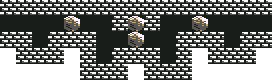}}\\
\subfloat[Configuration after a traversal from left to right.]{
	\label{fbd2}
	\includegraphics[scale=0.75]{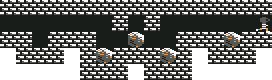}}\\
\subfloat[On the second traversal attempt, a boulder blocks the way.]{
	\label{fbd3}
	\includegraphics[scale=0.75]{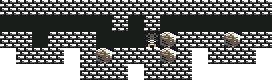}}
\caption{Single-use path for Boulder Dash.}
\label{fbd}
\end{figure}

\subsection{Deflektor (Vortex Software, 1987) is in \L}

\paragraph{Game description.}
This is a psychedelic puzzle game in which a laser ray has to be reflected around in the level by rotating several mirrors. The laser ray is emanated by a laser beam and must be reflected to a predefined location, after hitting several items that must be collected, while avoiding static mines, and without reflecting the laser back to the source for too long (which overheats the beam). All the relevant game elements are static, or can rotate in place. There are 16 possible orientations for the laser ray, and making it reach some location basically boils down to finding the correct way to orient the player-controlled mirrors. Some tiles act as reflecting walls, some are opaque and absorb the ray, some special tiles act as teleporters, others as self-rotating mirrors, or self-rotating \emph{polarizators} that may be traversed by the ray only in one direction at a time. All polarizators have eight possible orientations, and rotate at the same speed (or they are static). There are also some prisms that randomly refract the ray, and some gremlins that attach to player-controlled mirrors and randomly reorient them.

\paragraph{Complexity.}
This is a remarkable example of an ``easy'' commercial game: Deflektor is solvable in logarithmic space, which is a quite uncommon feature for a puzzle game, and possibly contributed to its modest success.

The key observation is that the ray never needs to be reflected twice by the same mirror in order to reach some location, because it can be re-oriented to any direction already on its first reflection.

For our purposes, prisms count as a special type of player-orientable mirror, as they effectively refract the ray in any desired direction and at the right moment, after waiting a long-enough time. Similarly, self-rotating mirrors count as regular mirrors, as the player can indeed slow them down or accelerate them. On the other hand, gremlins may be disregarded in our analysis, as their presence is merely a nuisance and never really prevents a level from being solved.

Next we show how to reduce Deflektor to the \L problem \UCON~\cite{papadimitriou,ucon}.
Recall that there are eight possible combined orientations of the polarizators, as they are either static or keep rotating at the same speed. Each combined orientation yields a \emph{reachability graph} $G_i$ on the game elements, which tells if a ray can be redirected by an object onto another. A reachability graph may be computed in \L by shooting the 16 possible rays from each mirror (or prism) and extending them until each ray is absorbed or reaches a relevant game element, such as a mirror or a collectible item. This necessarily happens after a finite number of reflections, because the possible ray slopes are rational, and a ray that is never absorbed must have a periodic trajectory.

Let $G^\star$ be the disjoint union of all the $G_i$'s, in which the eight copies of the laser beam are connected to a common \emph{beam vertex}. Finding a path in $G^\star$ from the beam vertex to one of the eight copies of an object means that the laser ray can be redirected to that object after suitably rotating the mirrors, and after waiting for the polarizers to be properly oriented.

The final graph is obtained as the disjoint union of several copies of $G^\star$, one for each item to collect. Let us arbitrarily order the items, and let $G^\star_j$ be the copy of $G^\star$ associated to the $j$-th item (here, the exit location counts as the last item in the list). Then, the eight copies of the $j$-th item in $G^\star_j$ are linked to the beam vertex of $G^\star_{j+1}$. The beam vertex of $G^\star_1$ will be called the \emph{starting vertex}, and the eight copies of the last item in the last copy of $G^\star$ are connected to a common vertex, which will be the \emph{ending vertex}.

Thus, the final graph has a path from the starting vertex to the ending vertex if and only if the mirrors can be oriented in such a way that the first item in the list can be collected, then reoriented to collect the second item, etc., and finally the exit location can be reached. Because items can be collected in any order, this is equivalent to solving the level.

\subsection{Lemmings (DMA Design, 1991) is \NP-hard}

\paragraph{Game description.}
The player has to guide a tribe of lemming creatures to safety through a hazardous landscape, by assigning them specific skills that modify their behavior in different ways. There is a limit to the number of times each skill can be assigned to a lemming, and skills range from building a short stair, to excavating a tunnel, to climbing vertical walls, etc. If no skill is assigned to a lemming, it keeps walking forward, turning around at walls, and falling into pits. Hazards include deadly pools of water or lava, several kinds of traps, and long falls. To beat a level, at least a given percentage of lemmings has to reach one of several exit portals within a time limit. The complete set of rules, especially the ways lemmings behave in different landscapes, is quite complex, and has been described by the author in~\cite{lemmings2}.

\paragraph{Complexity.}
The \NP-hardness of Lemmings was already proved by Cormode in~\cite{lemmings} using only Digger skills. More recently, in~\cite{lemmings2}, the author showed that Lemmings is \PSPACE-complete, even if there is only one lemming in the level, and only Basher and Builder skills are available (the technique used is based on a variant of Metatheorem~\ref{m2}.c). Here we propose a simple alternative proof of \NP-hardness, which uses only Basher skills (that allow lemmings to dig horizontally) and relies on Metatheorem~\ref{m1c}.b. Our construction can be easily modified to work with Miner skills, too (used to dig diagonally).

\begin{figure}[htbp]
\centering
\subfloat[Initial configuration with a trapped lemming.]{
	\label{flm1}
	\includegraphics[scale=0.7]{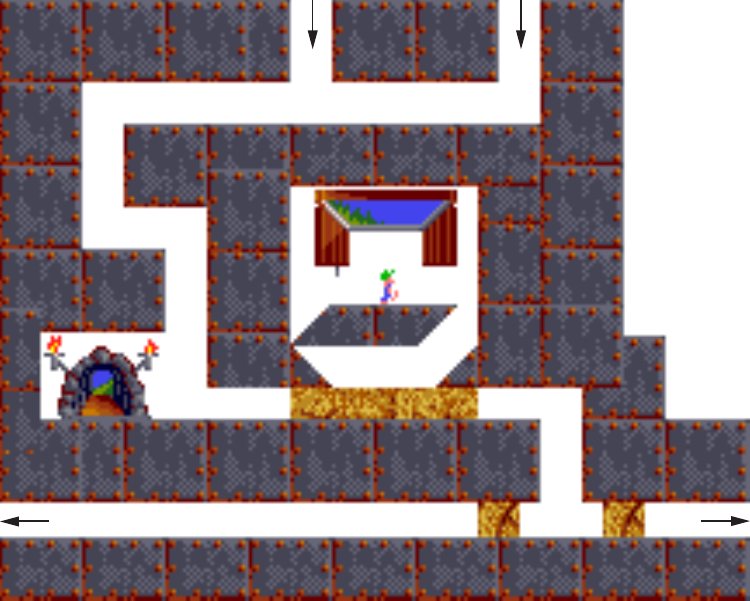}}\qquad
\subfloat[Setting the prisoner free with a Basher.]{
	\label{flm2}
	\includegraphics[scale=0.7]{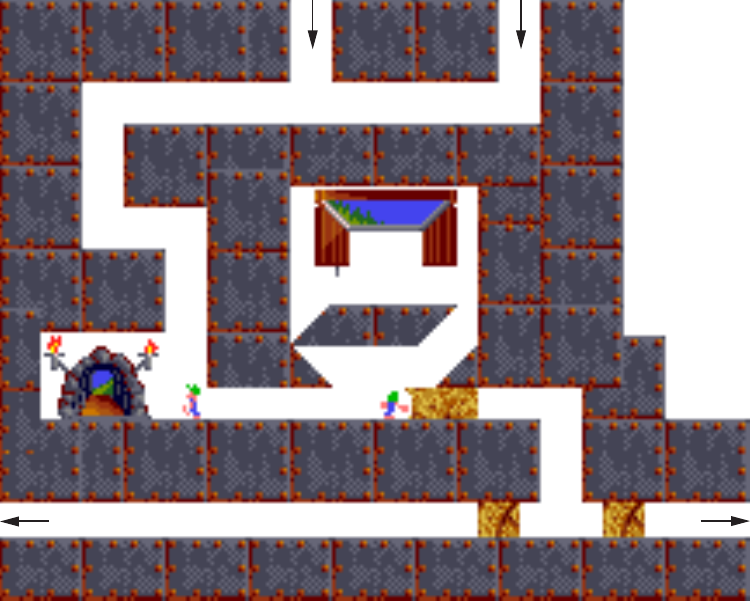}}
\caption{Enforcing location traversal in Lemmings.}
\label{flm}
\end{figure}

We model each location as in Figure~\ref{flm1}. As the game starts, exactly one lemming joins the level from the trapdoor, and is bound to stay in the enclosed area, walking back and forth. If either the indegree or the outdegree of the location is less than two, we suitably remove some of the passages marked by arrows. In the starting location, we fit a second trapdoor that releases another lemming in the upper corridor: this lemming is not a prisoner, and will be the avatar. In the final location, we replace the outgoing passages with an exit portal, which is intended for the avatar. We appropriately connect locations together as the arrows suggest, and we make sure that the right outgoing passage of the starting location leads to the exit location (rather than the left passage). It is quite simple to implement paths that go in any direction and can be effectively traversed by lemmings: simple ways to build them have been described in~\cite{lemmings} and~\cite{lemmings2}. All the lemmings in the level must reach an exit portal, and the initially available skills are $2n+2$ Bashers, where $n+1$ is the amount of locations in the level.

The tiles with the steel texture in Figure~\ref{flm} cannot be excavated, hence any lemming that is not the avatar is trapped inside a cage, waiting for the avatar to rescue it by Bashing the ground below, as Figure~\ref{flm2} illustrates. As a by-product of the lemmings' moving rules, when a prisoner hits the leftmost wall of its cage, it turns around and climbs on the upper platform. Then it falls down to the right, turns around at the wall, and walks on the lower ground from right to left. Hence, the prisoner is bound to come out of its cage facing left as soon as the avatar Bashes its ground, and it will inevitably reach the nearby exit portal. This implies that the avatar must visit every location in the level (location traversal) and that at least $n+1$ Basher skills have to be used to rescue all the prisoners. At any time, the number of available Basher skills minus the number of trapped lemmings will be understood as the number of keys carried by the avatar. Note that a small amount of ground (i.e., a door) must be Bashed in order to exit a location from an unvisited outgoing edge, and the initial amount of keys is $n+1$. This agrees with the key-door paradigm and the requirements of Metatheorem~\ref{m1c}.b, except for the presence of a door in the path between $v$ and $u$, which is matched by the extra available key (this must be the last door to be opened anyway, hence the discrepancy can be safely ignored).

The passages connecting locations are indeed one-way paths by construction, since lemmings cannot change direction unless they encounter a wall. Observe that the avatar may choose which outgoing path to take on its first traversal of a location. However, after the right path has been taken, there is no way to take the left path on the second traversal (but not vice versa). Fortunately, by inspecting the proof of Metatheorem~\ref{m1c}.b, this does not appear to be an issue, because it is not restrictive to assume that each location will be visited only once, except for the starting location, in which both outgoing paths must be taken. But we made sure that the right path leads to the exit location, and therefore it must be the last path to be traversed, which is indeed feasible.

\subsection{Lode Runner (Br\o derbund, 1983) is \NP-hard}

\paragraph{Game description.}
The player-controlled avatar must collect gold pieces while avoiding monsters, and is able to dig holes into certain floor tiles (those that look like bricks), which regenerate after a few seconds. Both the avatar and the monsters may fall into such holes, and the avatar cannot jump. The avatar is killed when it is caught by a monster, but it can safely stand in the tile directly above a monster's head. Monsters behave deterministically, according to the player's moves, although their behavior is often quite counterintuitive, as they do not always take the shortest path toward the player's avatar. When every gold piece has been collected, a ladder appears, leading to the exit.

\paragraph{Complexity.}
We apply Metatheorem~\ref{m1}: location traversal is implied by the collecting items feature, and a single-use path is illustrated in Figure~\ref{flr}. On the first traversal, the avatar can safely land on top of the monster and dig a hole to the left. The AI will make the monster fall in the hole, so the avatar may follow it, land on its top again, and proceed through a ladder, while the brick tile regenerates and the monster remains trapped in the hole below. The avatar cannot attempt to traverse the gadget a second time without getting stuck in the hole where the monster previously was (recall that the avatar can neither jump, nor dig holes horizontally).

\begin{figure}[htbp]
\centering
\subfloat[Initial configuration.]{
	\label{flr1}
	\includegraphics[scale=0.65]{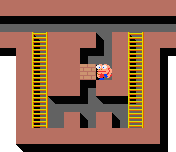}}\qquad\qquad
\subfloat[Safely landing on top of the monster.]{
	\label{flr2}
	\includegraphics[scale=0.65]{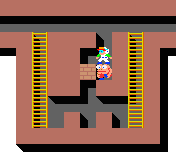}}\\
\subfloat[Digging the brick tile to set the monster free.]{
	\label{flr3}
	\includegraphics[scale=0.65]{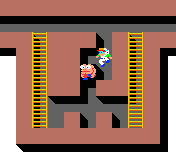}}\qquad\qquad
\subfloat[Following the monster down the hole and exiting through a ladder.]{
	\label{flr4}
	\includegraphics[scale=0.65]{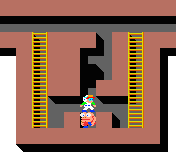}}\\
\subfloat[The brick tile regenerates, and the monster remains trapped below.]{
	\label{flr5}
	\includegraphics[scale=0.65]{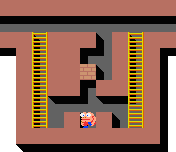}}\qquad\qquad
\subfloat[On the second traversal attempt, the avatar gets stuck.]{
	\label{flr6}
	\includegraphics[scale=0.65]{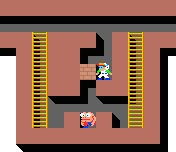}}
\caption{Single-use path for Lode Runner.}
\label{flr}
\end{figure}

\subsection{Mindbender (Magic Bytes, 1989) is \NL-hard}

\paragraph{Game description.}
This is the fantasy-themed sequel of Deflektor (see above), with a wizard shooting a ray of light in some direction, static gnomes holding orientable mirrors, kettles that must be collected by hitting them with the ray of light, and several new game elements. These include collectible keys that open locks, wandering monsters that eat kettles, movable blocks, etc.

All the elements of Deflektor are recreated in Mindbender, with substantially identical mechanics, with one crucial exception: polarizators in Mindbender are manually orientable by the player, whereas in Deflektor they rotate on their own.

\paragraph{Complexity.}
The full game is easily \NP-hard and arguably \PSPACE-complete, but the interesting fact is that even the subgame that is supposed to be ``isomorphic'' to Deflektor is in fact \NL-complete, thus harder than Deflektor.

\begin{figure}[htbp]
\centering
\subfloat[Initial configuration.]{
	\label{fmb1}
	\includegraphics[scale=0.85]{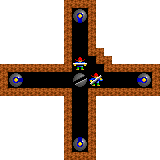}}\qquad\qquad
\subfloat[Redirecting a ray from left to top.]{
	\label{fmb2}
	\includegraphics[scale=0.85]{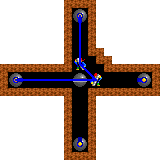}}
\caption{Implementing a directed graph in Mindbender.}
\label{fmb}
\end{figure}

We give a straightforward reduction from the \NL-complete problem \DCON~\cite{papadimitriou}: first of all, we may assume that each vertex of the given graph has indegree and outdegree at most two. Then, each such vertex is modeled with the gadget in Figure~\ref{fmb}. The left and bottom teleporters correspond to incoming edges, while the upper and right teleporters correspond to outgoing edges. Teleporters in different gadgets are connected together according to the topology of the given graph (which may be non-planar).

The central object is a polarizator, which can be oriented by the player and lets light rays pass in one direction only. The two gnomes can reflect the light either upward or rightward, as Figure~\ref{fmb2} exemplifies. No light ray can be redirected from the left teleporter to the bottom one, or vice versa, due to the polarizator.

The starting vertex contains the wizard instead of the left (or bottom) teleporter, and the ending vertex contains the exit door instead of the top (or right) teleporter. The level is solvable if and only if a path exists from the starting vertex to the ending vertex. Remarkably, no kettle has been used in the reduction.

\subsection{Pac-Man (Namco, 1980) is \NP-hard}

\paragraph{Game description.}
The player controls a yellow ball whose task is to collect all the \emph{pills} in a maze, while avoiding \emph{ghosts}. Collecting some special \emph{power pills} makes the avatar invulnerable for a short time, giving it the ability to temporarily disable ghosts upon contact. Despite this seeming simplicity, the full set of rules is quite complicated~\cite{pacman}, although just a few are relevant to our purposes.

Ghosts come in four different colors, and a ghost's color determines its behavior. However, all ghosts alternate between Chase mode and Scatter mode. In Chase mode, they follow the avatar with different heuristics, and in Scatter mode they head toward a preset location. There is also a Frightened mode, which is entered when the avatar collects a power pill, and makes all ghosts move randomly. After a few seconds, the effects of the power pill expire and all ghosts are back to Chase and Scatter modes. If a ghost in Frightened mode is touched by the avatar, it goes back to its starting location, a \emph{ghost house}, and comes out again shortly.

As a general rule, every time there is a mode switch, all ghosts immediately reverse their direction. Other than that, ghosts may never reverse direction, not even upon reaching a maze intersection, and not even when in Frightened mode. (This is also a practical way for the player to tell when a mode switch occurs.)

Depending on the game level, all timings and speeds are subject to variations: ghosts may be faster or slower in different modes, and the durations of the three modes may vary. Usually, during Frightened mode, the avatar speeds up and the ghosts slow down.

\paragraph{Complexity.}
The decision problem is whether a level can be completed without losing lives. We assume full configurability of the amount of ghosts and ghost houses, speeds, and the durations of Chase, Scatter, and Frightened modes. We do not alter the basic game mechanics or the AI, though.

We prove \NP-hardness by applying Metatheorem~\ref{m1b}.a. A location with an adjacent toll road is sketched in Figure~\ref{fpa}. Power pills are used to model tokens, so the starting location contains two power pills, and the final location contains none. Hence, to properly enforce location traversal, we further place a normal pill in the final location.

\begin{figure}[htbp]
\centering
\subfloat[Initial configuration, approached from the top.]{
	\label{fpa1}
	\includegraphics[scale=0.65]{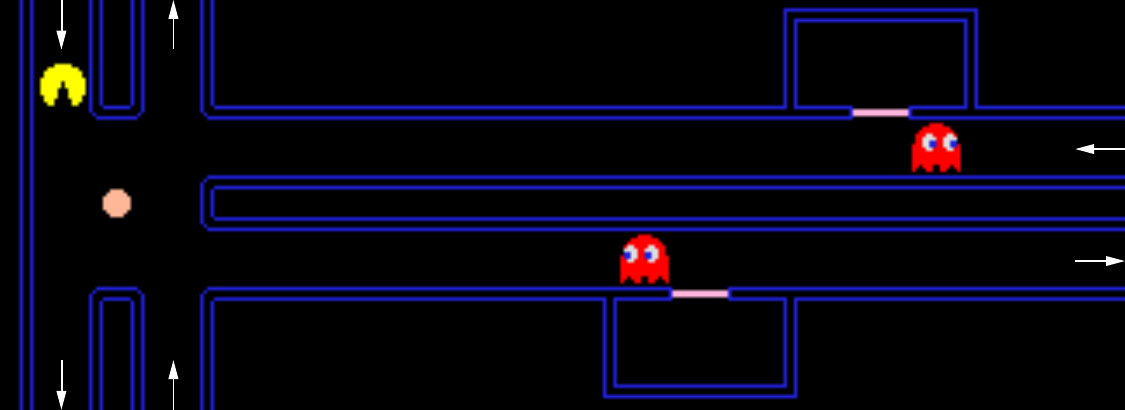}}\\
\subfloat[Collecting the power pill to traverse the corridor.]{
	\label{fpa2}
	\includegraphics[scale=0.65]{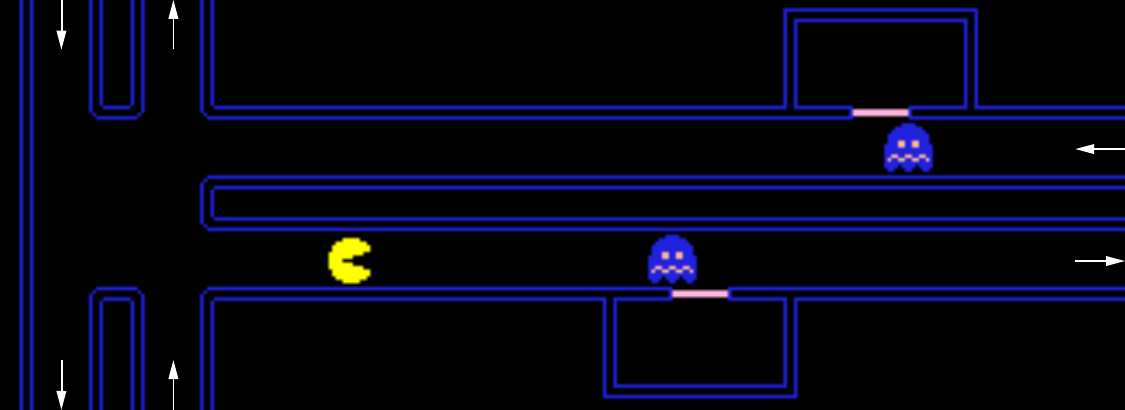}}\\
\subfloat[Exiting to the right while the ghosts recover.]{
	\label{fpa3}
	\includegraphics[scale=0.65]{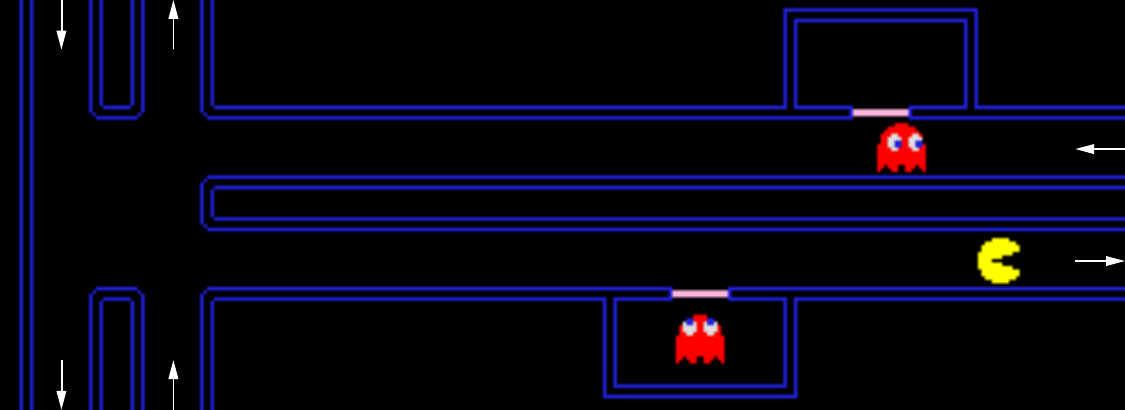}}
\caption{Sketch of a location and toll road for Pac-Man.}
\label{fpa}
\end{figure}

Each toll road is implemented as a pair of parallel maze corridors, each of which contains a ghost house somewhere, spawning one red ghost. The two corridors are intended to be traversed in opposite directions by the avatar (i.e., they are one-way paths).

Chase and Scatter modes have the same duration, and all ghosts have the same speed in both modes. Let $p$ be the number of tiles each ghost covers between two mode switches. Frightened mode lasts longer, but ghosts slow down, covering exactly $f$ tiles while in that mode. We make sure that each ghost house is found exactly $d=p+(n+1)f$ tiles away from its corridor's entrance, and that each corridor is more than $2d$ tiles long.

As the game starts, the ghosts spawn in front of their respective ghost houses, and start in Chase mode, following the corridor in some direction. Whenever a mode switch occurs, all ghosts reverse direction, and they cannot change it again until the next mode switch, because they never reach a maze intersection. As a result, each ghost ``patrols'' a portion of length $p$ of its own corridor. By construction, since Frightened mode can be entered at most $n+1$ times, no ghost may ever leave its corridor.

Upon collecting a power pill, the avatar's speed increases in such a way that it can cover $2d$ tiles (or slightly more) into any adjacent corridor. By doing so, the avatar consumes a token and, if the corridor is traversed in the proper direction, the ghost is necessarily encountered and sent back to the ghost house. By the time the avatar has reached the end of the corridor, the power pill's effects expire and the ghost comes out of the ghost house, making the toll road functional again.

\subsection{Pipe Mania (The Assembly Line, 1989) is \NP-complete}

\paragraph{Game description.}
Not to be confused with KPlumber, with a similar theme but much different mechanics~\cite{survey1}, in this puzzle game a long-enough pipe has to be constructed out of several pieces, randomly presented in a queue, starting from a given \emph{source location}. After a timer expires, a stream of water starts flowing from the source into the pipes, and the game is won if and only if the stream traverses a given number of tiles before spilling out.

Since the player can keep constructing pipes on the same tile, ``overwriting'' the previous pieces until he gets the piece that he wants, he may indeed shape the pipe as he pleases, if the initial timer lasts long enough. Some obstacles are also present in each level, such as fire hydrants, on which pipes cannot be built.

\paragraph{Complexity.}
Membership in \NP is obvious. For \NP-hardness, we apply Metatheorem~\ref{m1}. We use obstacles to model the boundaries of locations and paths, as Figures~\ref{fpm1} and~\ref{fpm3} illustrate. The resulting paths are necessarily single-use, as only one pipe can fit in them.

\begin{figure}[htbp]
\centering
\subfloat[Sketch of a location.]{
	\label{fpm1}
	\includegraphics[scale=0.75]{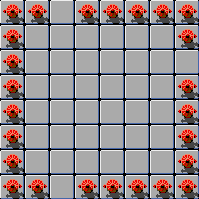}}\qquad\quad
\subfloat[Building a pipe that uses most tiles twice.]{
	\label{fpm2}
	\includegraphics[scale=0.75]{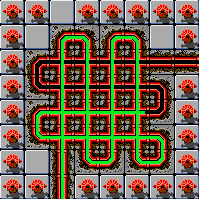}}
\caption{Enforcing location traversal in Pipe Mania.}
\label{fpma}
\end{figure}

We still need to establish location traversal. Suppose we implemented our planar graph with orthogonal lines as edges and squares as vertices. Let $\ell$ be the total length of the paths plus the area of the starting vertex, and $a$ be the side length of a generic vertex. If the number of vertices is $n+1$, the number of paths is $\Theta(n)$, because the graph is 3-regular (refer to the proof of Metatheorem~\ref{m1}).

\begin{figure}[htbp]
\centering
\subfloat[Initial configuration.]{
	\label{fpm3}
	\includegraphics[scale=0.75]{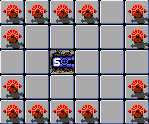}}\qquad\qquad\qquad
\subfloat[Traversal example.]{
	\label{fpm4}
	\includegraphics[scale=0.75]{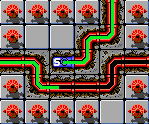}}
\caption{Starting location in Pipe Mania.}
\label{fpmb}
\end{figure}

Imagine scaling our construction by an integer factor $k$, in such a way that all paths preserve their unit width, but just increase their length. All the vertices are also scaled in size, except the starting vertex, which remains constant. Then, the total length of all paths plus the starting vertex area becomes $l'= kl-|\Theta(n)|$, and the area of a non-starting vertex becomes $A = k^2a^2 - |\Theta(ka)|$. When the pipe reaches a generic vertex, it can cover most of its tiles twice before taking the path toward another vertex (cross-shaped pieces must be used, see Figure~\ref{fpm2}). The length of such a pipe is at least $A'=2A-|\Theta(ka)|=2k^2a^2 - |\Theta(ka)|$. Let us set $k$ to a suitable $\Theta(n)$, so that $l'$ becomes negligible compared to $A'$. Now, it is sufficient to set the required length of the pipe to $nA'$ to ensure that all the vertices will be covered by it.

Recall that the starting vertex must be traversed twice by the pipe: Figure~\ref{fpm4} shows an example of how this can be done.

\subsection{Prince of Persia (Br\o derbund, 1989) is \PSPACE-complete}

\paragraph{Game description.}
The player has to guide an avatar through several dungeon levels, opening gates, fighting guards, and avoiding traps. Most gates are operated by pressure plates. The avatar can walk, run, jump, climb, duck, fight, etc.

\paragraph{Complexity.}
The game's \PSPACE-hardness was first proved in~\cite{platform}, but the rather involved construction may be replaced by a somewhat simpler one based on Metatheorem~\ref{m2}.c, which in addition does not rely on gravity, long falls, or on doors that can be opened by more than one pressure plate. To prevent the avatar from avoiding a pressure plate by jumping past it, we simply put it on an elevated tile, which has to be climbed in order to be traversed, as Figure~\ref{fpp} shows. We can even do without vertical walls (as in~\cite{platform}), because they can be substituted with unopenable gates.

\begin{figure}[htbp]
\centering
\includegraphics[scale=1.5]{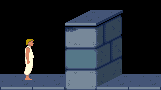}
\caption{Unavoidable pressure plate in Prince of Persia.}
\label{fpp}
\end{figure}

Membership in \PSPACE $=$ \NPSPACE (cf.~Savitch's theorem~\cite{papadimitriou}) is quite obvious, as the whole level's configuration can be stored in linear space, and enemy guards have a very simple pseudo-random fighting pattern.

\subsection{Puzzle Bobble 3 (Taito, 1996) is \NP-complete}

\paragraph{Game description.}
In this Tetris-like puzzle game, levels are made of several colored bubbles, stacked in a hexagonal distribution. The player controls a cannon at the bottom of the screen, which can shoot new bubbles of random colors in any direction. Bubbles attach to each other and, whenever at least three monochromatic bubbles form a connected set as a result of a shot, they pop. Monochromatic triplets may indeed be present in the initial level configuration, and they pop only when hit by a new bubble of the same color.

Some \emph{anchors} hold the whole stack together and, as soon as a bubble is not in the same ``connected component'' with an anchor, it falls down and is eliminated.

Apart from colored bubbles, there are \emph{stone blocks} that cannot be popped (but may fall if not held up by an anchor), and \emph{rainbow bubbles} that turn the same color of any bubble that pops next to them, and can later be popped like normal bubbles. Notably, if a set of at least three adjacent monochromatic bubbles is formed as a result of some rainbow bubbles turning that color, they immediately pop. This may even induce a ``chain reaction'' of exploding rainbow bubbles, during which the player is not allowed to shoot a new bubble, and must wait for the explosions to finish.

The goal is to clear all anchors, and an anchor is cleared if and only if no bubble is attached to it.

\paragraph{Complexity.}
We prove \NP-hardness by a reduction from \PTSAT. Several variable gadgets (Figure~\ref{fpba}) are stacked on top of each other, slightly staggered, on the far left of the construction. The clause gadgets (Figure~\ref{fpbb}) are on the right, far above the variable gadgets. To separate \emph{variable layers} from each other and from the clause gadgets, we put long \emph{shields} of stone blocks, extending from each variable gadget to the far right of the construction. The last shield (i.e., the one in the top layer) also extends all around the whole construction, on the right, top and left sides, preventing bubbles shot by the player from bouncing on the sides of the screen. Variables and clauses are connected via carefully shaped \emph{fuses} made of rainbow bubbles, forking and bending as in Figure~\ref{fpb1}.

\begin{figure}[htbp]
\centering
\subfloat[Initial configuration.]{
	\label{fpb1}
	\includegraphics[scale=0.65]{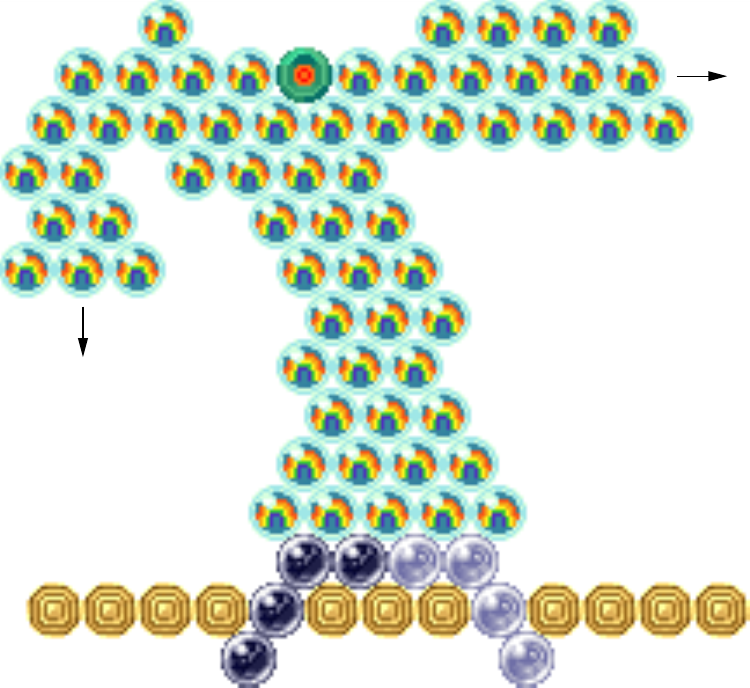}}\qquad\quad
\subfloat[Setting the variable to true. Part of the shield falls.]{
	\label{fpb2}
	\includegraphics[scale=0.65]{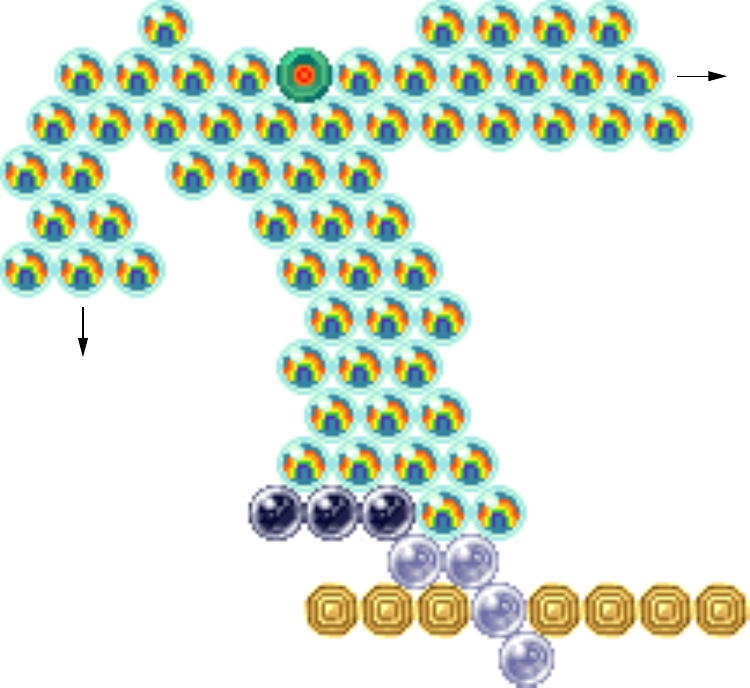}}\\
\subfloat[The fuse is consumed and the shield completely falls.]{
	\label{fpb3}
	\includegraphics[scale=0.65]{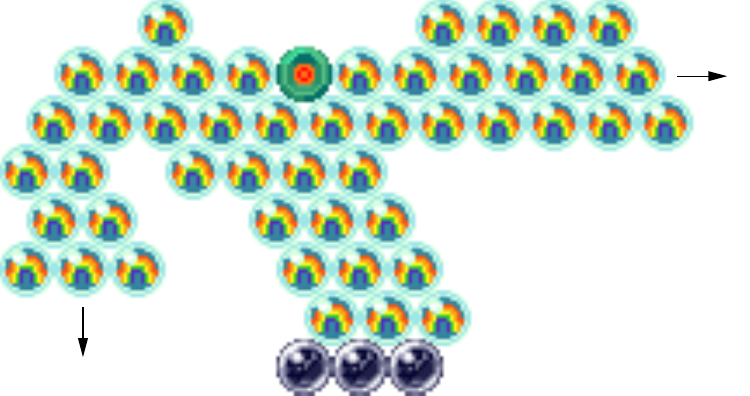}}\qquad\quad
\subfloat[The signal splits and the anchor is eliminated.]{
	\label{fpb4}
	\includegraphics[scale=0.65]{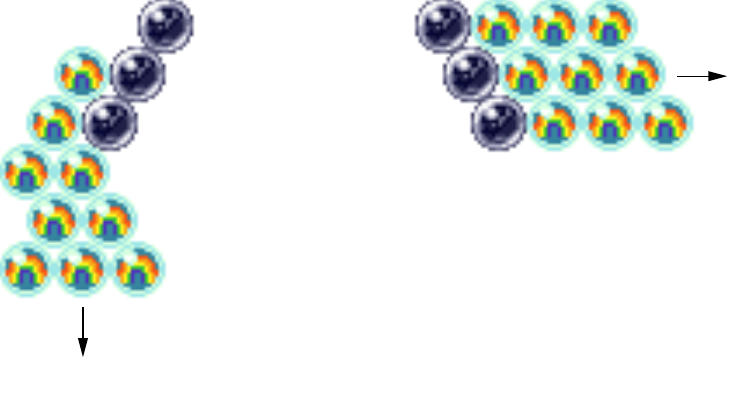}}
\caption{Variable gadget for Puzzle Bobble 3, exemplifying a fork of the fuse.}
\label{fpba}
\end{figure}

Initially, only the bottom variable gadget is exposed, and the player may choose whether to pop the black or the white bubbles, which correspond to opposite truth values. Popping one of the two sets ---say, the black one--- causes three rainbow bubbles to turn black and pop immediately after. This triggers a chain reaction, in which at least three new rainbow bubbles turn black and pop at each step, consuming the fuse and eventually reaching the clause gadgets. At this point, a thin colored \emph{wire} is reached in every clause gadget (see Figure~\ref{fpbb}), which pops if and only if it is black (its color tells whether the corresponding literal in the clause is positive or negative). If it pops, the explosion propagates inside the clause gadget, eliminating the anchor. Notice that the explosion can never ``backfire'' from the clause gadget and consume fuses corresponding to different variables, because each wire is connected to only two rainbow bubbles of its attached fuse.

\begin{figure}[htbp]
\centering
\subfloat[Initial configuration.]{
	\label{fpb5}
	\includegraphics[scale=0.65]{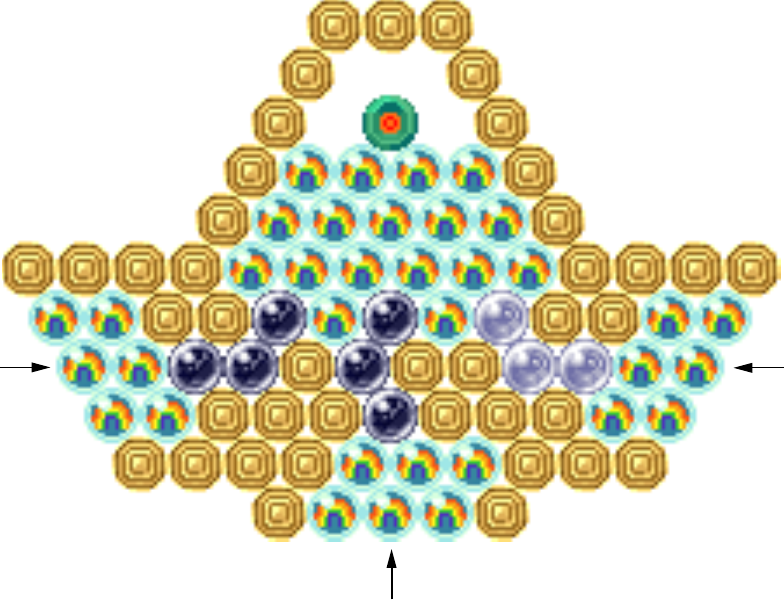}}\qquad\quad
\subfloat[$x$ is set to true.]{
	\label{fpb6}
	\includegraphics[scale=0.65]{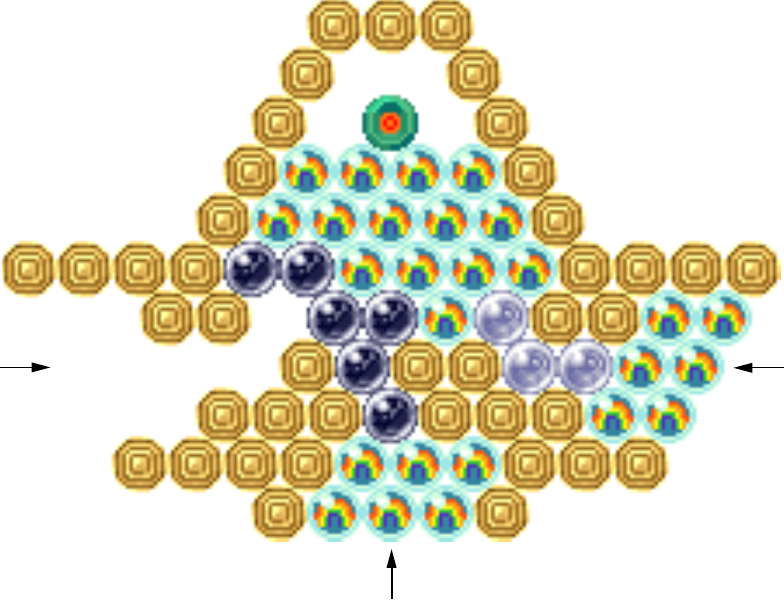}}\\
\subfloat[The signal does not propagate in the $y$ wire.]{
	\label{fpb7}
	\includegraphics[scale=0.65]{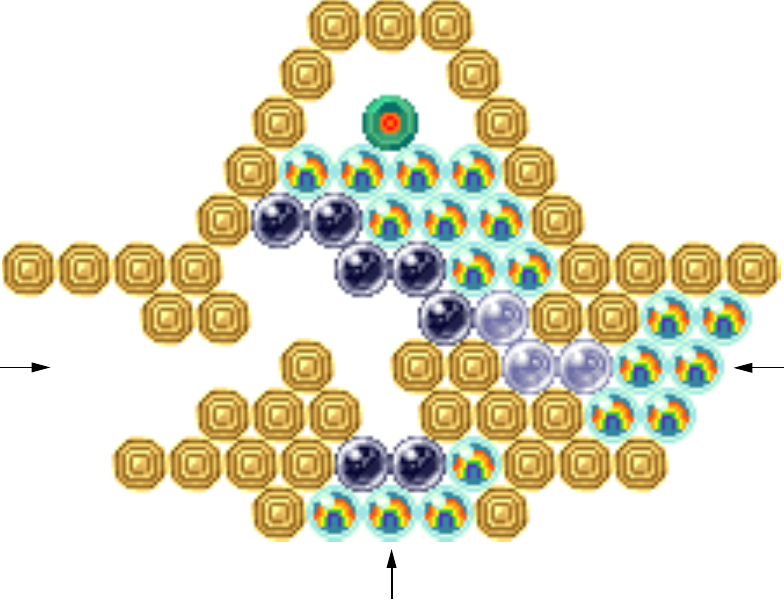}}\qquad\quad
\subfloat[The anchor is eliminated.]{
	\label{fpb8}
	\includegraphics[scale=0.65]{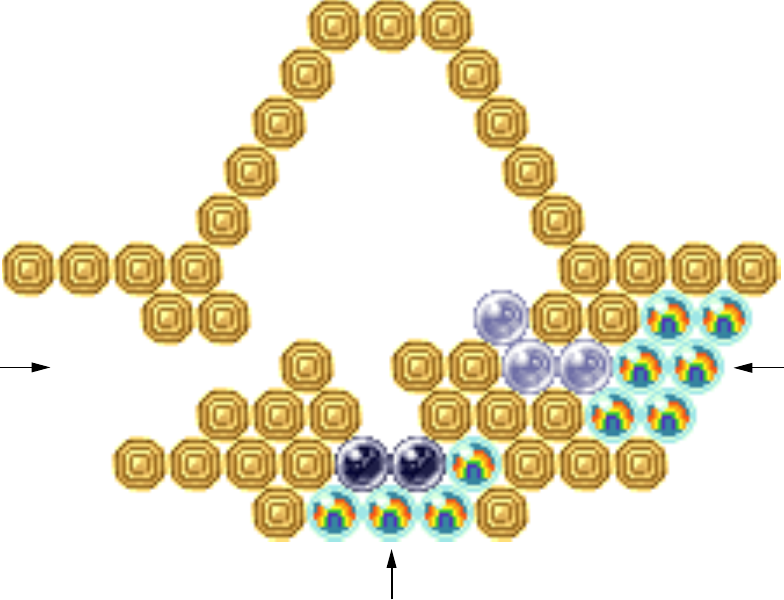}}
\caption{Clause gadget for Puzzle Bobble 3, corresponding to $(x \vee y \vee \neg z)$.}
\label{fpbb}
\end{figure}

When the fuse of the first variable has been consumed, the remaining part of the variable layer falls, including the shield (see Figure~\ref{fpba}). The second variable layer is then exposed, and the process continues until all fuses have been consumed, and all shields have fallen. What eventually remains are the ``unsatisfied'' clause gadgets, whose wires are now impossible to reach, due to the surrounding \emph{sheaths} made of stone blocks. Notice that each variable layer has its own anchor, so its shield does not fall until the variable has been set by the player, even if all the clauses connected to that variable have already been satisfied.

This proves \NP-hardness. Completeness holds under the assumption that the player can always choose the color of his next bubble, which is not far from true in most cases, since bubbles can be either discarded by making them bounce back to the bottom of the screen, or can be stacked somewhere (if done properly, not more than two bubbles per color need to be stacked at once).

\subsection{Skweek (Loriciels, 1989) is \NP-hard}

\paragraph{Game description.}
The player controls a furry ball that has to walk on blue tiles in order to paint them pink, while avoiding monsters. Some tiles are made of ice and do not have to be painted, the avatar slides on them and is unable to change direction until it reaches a different type of tile, or its slide is blocked by a wall. Some blue tiles fall apart when the avatar steps on them, opening a hole in the ground that becomes a deadly area. All the blue tiles have to be painted pink within a time limit in order to finish the level. Several power-ups randomly appear, including an exit door and teddy bears of several colors, which let the player immediately skip the level when collected.

\paragraph{Complexity.}
The decision problem is whether a given level can be completed without losing lives, regardless of the power-ups that may randomly appear. The presence of breakable tiles yields an immediate application of Metatheorem~\ref{m1}. Figure~\ref{fsk} shows how a location is constructed: location traversal is implied by the blue tiles, all of which have to be covered by the avatar. On the other hand, after traversing a path connecting two locations, the cracked tiles break and cannot be accessed again, making it a single-use path. The reason why we use ice tiles is merely that they need not be painted, so the player's purpose is in fact to visit all locations, rather than all paths. For this reason, even though there exists a power-up that prevents the avatar from sliding on ice, our construction still works as intended.

\begin{figure}[htbp]
\centering
\subfloat[Initial configuration, with the avatar approaching from the top path.]{
	\label{fsk1}
	\includegraphics[scale=0.65]{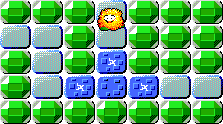}}\qquad\quad
\subfloat[Breaking the top tile and painting the central tile.]{
	\label{fsk2}
	\includegraphics[scale=0.65]{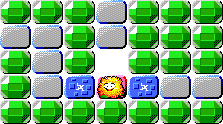}}\\
\subfloat[Reaching the left tile and running back while it breaks.]{
	\label{fsk3}
	\includegraphics[scale=0.65]{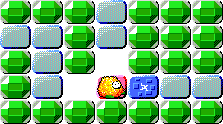}}\qquad\quad
\subfloat[Breaking the last tile and exiting to the right, leaving no blue tile behind.]{
	\label{fsk4}
	\includegraphics[scale=0.65]{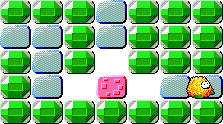}}
\caption{Implementing location traversal and single-use paths in Skweek.}
\label{fsk}
\end{figure}

Proving membership in \NP would be almost straightforward, were it not for a certain type of monster that turns tiles from pink to blue. On top of this, monster behavior is pseudo-random and partly depends on the player's moves. For these reasons, beating a level may conceivably take an exponentially long time, and any proof that the game lies in \NP would have to be carefully crafted.

\subsection{Starcraft (Blizzard Entertainment, 1998) is \NP-hard}

\paragraph{Game description.}
Starcraft is a real-time strategy game in which two or more players have to train an army in order to destroy each other's bases. Two types of resources can be gathered from the environment by special units called \emph{workers}. Resources allow to make new buildings in order to train more units, thus forming an army that can be sent to war. There are three possible \emph{races} to choose, each of which has its unique unit types, each one with different parameters, such as hit points, range, damage, speed, etc. Some units have special abilities, such as becoming invisible, casting offensive or defensive ``spells'', etc. A player loses if and only if all his buildings are destroyed, regardless of the amount of units that he still has, or the amount of resources he gathered.

\paragraph{Complexity.}
Most RTS games are expected to be \EXP-hard, since they involve at least two players, and a match may last an arbitrarily long time. However, a simple \NP-hardness proof can be given via Metatheorem~\ref{m1b}.a. The same reasoning applies, with minor changes, to several RTS games other than Starcraft, such as Warcraft and Age of Empires.

We produce a configuration in a Protoss vs.\ Terran game, in which deciding whether the Protoss player can win or the game is a draw is \NP-hard.

In our setting, the avatar will be a Protoss Probe (i.e., the Protoss race's worker unit). When several Probes will be found in the same area, we will identify one as the avatar, and all the other Probes will represent tokens carried by the avatar.

\begin{figure}[htbp]
\centering
\includegraphics[scale=0.35]{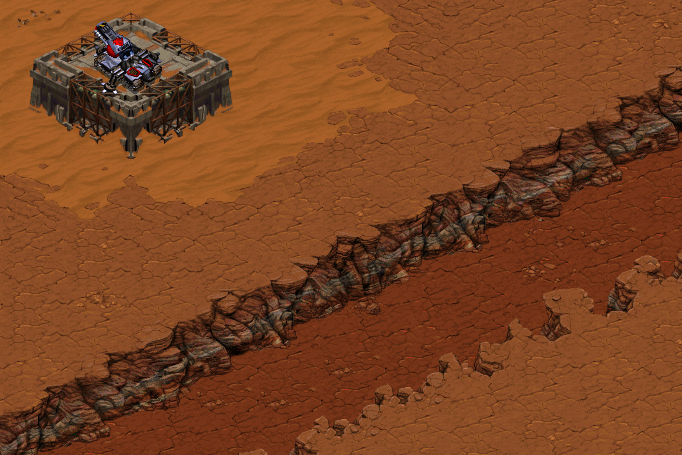}
\caption{Toll road for Starcraft.}
\label{fsc1}
\end{figure}

Consistently with our avatar-token abstraction, we present a toll road in Figure~\ref{fsc1}. If a lone Probe attempts to walk in the canyon, it is destroyed by a single shot of the enemy Siege Tank. But if two Probes walk together, only one can be targeted and destroyed, whereas the second Probe can make it past the Siege Tank while it reloads (the two Probes should stay a couple of tiles away from each other, to avoid splash damage). Paraphrasing, the avatar can traverse the toll road if and only if it is carrying a token. It does not matter which Probe is destroyed, because they are all equivalent. In general, if several Probes attempt to traverse the canyon together, at least one is destroyed, and at least one survives.

Figure~\ref{fsc2} shows how to implement a token lying in some location. The building is a Protoss Nexus, and there is a Probe trapped behind a Mineral Field, worth exactly $100$ Minerals. The Probe can gather up to eight Minerals at once, but then it must bring them to a Nexus before it can gather more Minerals. It follows that the trapped Probe cannot free itself, but it must wait for another Probe to set it free by bringing all the Minerals to the Nexus.

\begin{figure}[htbp]
\centering
\includegraphics[scale=0.35]{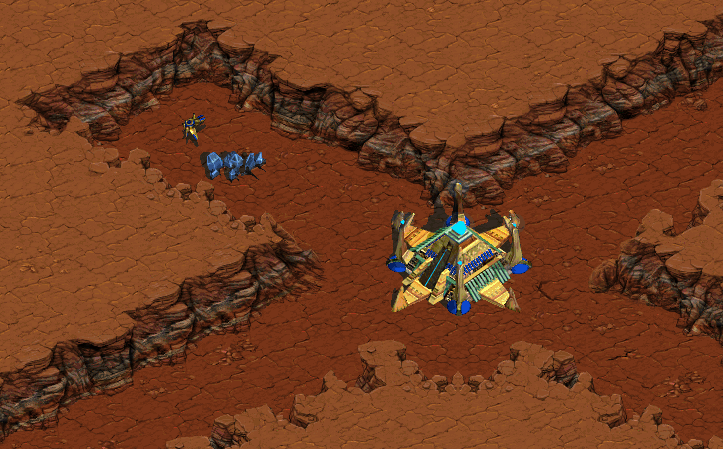}
\caption{Implementing tokens in Starcraft.}
\label{fsc2}
\end{figure}

In our analogy, gathering all the Minerals in a location to set a Probe free corresponds to picking up a token. If now we add a free Probe to the designated starting location, we are effectively placing an avatar there, which must set a new Probe free every time it needs to traverse a toll road.

We still have to enforce location traversal and ensure that no new Probes are trained. Recall from the proof of Metatheorem~\ref{m1b}.a that there are $n-1$ locations with one token, plus a starting location with two tokens, and a final location with no tokens. Therefore, in our generated map, the starting location has two Mineral Fields, and the final location has none. We place a Nexus and a Mineral Field in the final location as well (but no Probe, i.e., no token), so that there are $100(n+2)$ Minerals in total, and at least $100$ Minerals in each location.

In a different part of the map, we place $n+2$ copies of the ``crater'' shown in Figure~\ref{fsc3}. In each crater, there is a Protoss Gateway supported by a Protoss Pylon, and a Terran Supply Depot. This completes our construction.

\begin{figure}[htbp]
\centering
\includegraphics[scale=0.45]{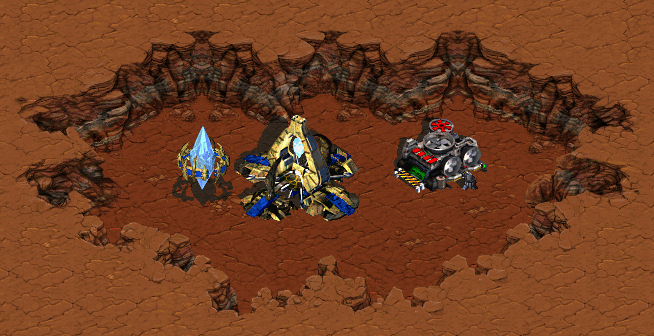}
\caption{A Zealot must be trained to destroy the Supply Depot.}
\label{fsc3}
\end{figure}

It is clear that the Terran player has no way to gather resources or train new units, in that he has no workers, and his only buildings are $n+2$ Supply Depots. Moreover, each Siege Tank is bound to stay on a small platform, and switching from Siege mode to Tank mode has no use, because the shortened attack range would not allow it to hit any Probe.

On the other hand, let us assume that the Protoss player starts with no resources, either. Because there are no Vespene Geysers on the map, and only Minerals are available, there are only two kinds of units that the Protoss player may ever train: Probes from Nexuses and Zealots from Gateways. Because none of them can fly and both are melee units, it follows that a Zealot must be trained from each Gateway in order to destroy the nearby Supply Depot. Since training a Zealot costs 100 Minerals, all the Minerals in every location must be gathered and spent just for training Zealots. Therefore, no additional buildings may be built, and no additional Probes may be trained at the Nexuses. In other terms, just the initial avatar can be used, and all the locations must be reached, which implies the location traversal feature.

\subsection{Tron (Bally Midway, 1982) is \NP-hard}

\paragraph{Game description.}
There are four subgames, one of which is a ``light cycle'' race between the player and several opponents. The race takes place in a rectangular grid whose external boundary is a deadly obstacle. The trail of each light cycle becomes a deadly obstacle as well, hence the safe areas become narrower and narrower as the race progresses. As soon as a light cycle hits an obstacle, it is eliminated, and also its trail is removed from the grid. The goal is to remain the sole survivor in the arena.

\paragraph{Complexity.}
This game becomes \PSPACE-complete if played on abstract graphs~\cite{tron} whereas, for the standard plane grid version, a simple \NP-hardness proof can be given, as an application of Metatheorem~\ref{m1}.

First we properly embed the location-path structure into the grid, and then we show how to actually implement each component with light cycle trails.

We start from a grid-aligned embedding in which each vertex is a square, and paths have unit width. Then we scale the construction up by some large-enough factor, while preserving the size of the vertices, and keeping all the paths of unit width. We do so to make the total area of all vertices negligible compared to the area of a face of the underlying plane graph.

Next we perform the same operation that we did for Pipe Mania (see above): with the same notation, we scale the construction by a factor $k=\Theta(n)$, so that the resulting combined path length $l'$ is negligible with respect to the area of a single vertex $A$. (In contrast with Pipe Mania, though, the starting vertex counts as a regular vertex and its area does not contribute to $l$ or $l'$.)

\begin{figure}[htbp]
\centering
\includegraphics[scale=0.45]{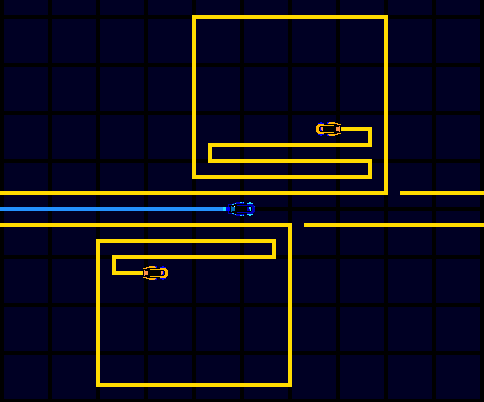}
\caption{Sketch of a single-use path for Tron.}
\label{ftr}
\end{figure}

Then we proceed by implementing paths and locations, as sketched in Figure~\ref{ftr}. Each opponent light cycle is responsible for drawing the border of a face of the plane graph underlying our construction (including the outer face), which is a grid-aligned polygon. When a light cycle is done drawing and meets its own trail again, it turns around and ``traps'' itself in a rectangle of area $(n+1)A$ (or slightly smaller) inside the polygon it just outlined. This rectangle necessarily fits somewhere in the polygonal face, by the first step of the above construction. In Figure~\ref{ftr} we see a path, traversed by the player's light cycle, which is bordered by two faces, the upper face arguably having a smaller perimeter than the lower one (because the upper rectangle is bigger, and a larger part of it has been covered).

While paths and vertices are constructed, we assume that the player ``waits'' by covering a small square in the starting vertex. This is feasible, because the perimeter of any face is much smaller than the area of a vertex, by construction. Then the actual race starts, and the player has to cover enough locations to survive longer than his opponents. Paths are obviously single-use, because they have unit width, and the player's trail is an obstacle even for the player himself. Location traversal is implied by the fact that the player's light cycle must cover at least a length of slightly less than $(n+1)A$, so it must visit all vertices.

\end{document}